\documentclass[12pt]{article}

\usepackage[letterpaper,
            top=1in, bottom=1in,
            left=1in, right=1in,
            headheight=14pt]{geometry}
\usepackage[T1]{fontenc}

\usepackage{amsmath,amssymb,amsthm}
\newtheorem{theorem}{Theorem}

\usepackage{graphicx}
\usepackage{subcaption}
\usepackage{tikz}
\usetikzlibrary{arrows.meta,positioning,shapes.geometric}

\usepackage{booktabs}
\usepackage{makecell}
\usepackage{array}
\usepackage[table]{xcolor}
\definecolor{HeaderBlue}{RGB}{220,230,245}
\definecolor{BestGreen}{RGB}{220,245,220}

\usepackage{algorithm}
\usepackage{algorithmic}

\usepackage{xcolor}
\usepackage{url}
\usepackage[colorlinks=true,linkcolor=blue,citecolor=blue,urlcolor=blue]{hyperref}
\usepackage{cite}
\usepackage{fancyhdr}
\usepackage{parskip}
\usepackage{titlesec}
\titleformat{\section}{\large\bfseries}{\thesection.}{0.5em}{}
\titleformat{\subsection}{\normalsize\bfseries}{\thesubsection.}{0.5em}{}
\titleformat{\subsubsection}{\normalsize\itshape}{\thesubsubsection.}{0.5em}{}
\titlespacing*{\section}{0pt}{14pt}{6pt}
\titlespacing*{\subsection}{0pt}{10pt}{4pt}
\titlespacing*{\subsubsection}{0pt}{8pt}{2pt}

\fancypagestyle{plain}{%
  \fancyhf{}%
  \fancyfoot[C]{\footnotesize \copyright~2026 IEEE. A two-page version of this work was published in the
    \textit{2026 IEEE 23rd Consumer Communications \& Networking Conference (CCNC)}.
    DOI: \href{https://doi.org/10.1109/CCNC65079.2026.11366526}{10.1109/CCNC65079.2026.11366526}}%
}

\begin{document}

%% ---------- Title block ----------
\begin{center}
  {\LARGE\bfseries Statistical Analysis and Optimization of the MFA\\[6pt]
  Protecting Private Keys}

  \vspace{14pt}

  {\large
    Mahafujul Alam \quad Julie B.\ Heynssens \quad Bertrand Francis Cambou
  }

  \vspace{6pt}

  {\normalsize
    School of Informatics, Computing and Cyber Systems,
    Northern Arizona University, Flagstaff, AZ 86011, USA\\
    \texttt{ma3755@nau.edu} \enspace
    \texttt{Julie.Heynssens@nau.edu} \enspace
    \texttt{Bertrand.Cambou@nau.edu}
  }

  \vspace{6pt}

  {\small\itshape
    \copyright~2026 IEEE. A two-page version of this paper was published in the
    \textit{2026 IEEE 23rd Consumer Communications \& Networking Conference (CCNC)}.\\
    DOI: \href{https://doi.org/10.1109/CCNC65079.2026.11366526}{10.1109/CCNC65079.2026.11366526}
  }
\end{center}

\thispagestyle{plain}
\vspace{4pt}
\hrule
\vspace{8pt}

%% ---------- Abstract ----------
\begin{center}
  \textbf{Abstract}
\end{center}
\noindent
In the current information age, asymmetrical cryptography is widely used to protect
information and financial transactions such as cryptocurrencies. The loss of private keys can
have catastrophic consequences; therefore, effective MFA schemes are needed. In this paper,
we focus on generating ephemeral keys to protect private keys. We propose a novel
bit-truncation method in which the most significant bits (MSBs) of response values derived
from facial features in a template-less biometric scheme are removed, significantly improving
both accuracy and security. A statistical analysis is presented to optimize an MFA comprising
at least three factors: template-less biometrics, an SRAM PUF-based token, and passwords.
The results show a reduction in both false-reject and false-acceptance rates, and the
generation of error-free ephemeral keys.

\vspace{6pt}

\vspace{6pt}
\hrule
\vspace{6pt}

%% ============================================================
\section{Introduction}
\label{sec:intro}

Asymmetric cryptography has been successfully implemented in a wide range of applications
vital to the information age, such as cryptocurrencies, secure mail, public key
infrastructure, and credit card protection. In distributed cryptocurrency applications,
private keys are used to sign all transactions, while public keys are required to verify
them. To strengthen this process, Multi-Factor Authentication (MFA) schemes have been
proposed to protect private keys with ephemeral keys~\cite{ChallengeResponsePair2025a},
with the objective of operating in secure zero-trust distributed environments while also
delivering zero-knowledge MFA. Here, the term zero-knowledge MFA is defined as an
authentication method in which all factors are tested concurrently, thereby preventing
serial attacks against individual factors. The development of such a protection scheme is
considerably more demanding than conventional MFA schemes, which typically deliver only
go/no-go authentication. In contrast, the ephemeral keys generated in our scheme cannot
contain a single erroneous bit, and any change in one factor must invalidate the transaction.

In this work, we propose a bit-chopping technique applied to template-less biometrics. By
removing the most significant bits (MSBs) from distance measurements between facial
landmarks and challenge points, coarse variations that lead to false acceptances are
eliminated, while subject-specific detail is retained. We conduct a statistical analysis of
this scheme by varying ADC precision and the number of chopped MSBs, evaluating its impact
on false-acceptance, false-reject, and bit-error rates.

We further analyze an
SRAM-PUF~\cite{SiliconPhysicalRandom2002,herderPhysicalUnclonableFunctions2014}-based
token, focusing on the number of enrollment cycles required to identify unstable cells. The
objective is to determine the optimal power-on-off cycle count that minimizes bit errors
while keeping the enrollment process efficient. Together, the optimized biometric scheme
with MSB bit-chopping and the SRAM-PUF scheme form the basis of a zero-knowledge MFA
protocol that generates stable, error-free ephemeral keys for protecting the private keys.

%% ============================================================
\section{Related Work}
\label{sec:related}

In Challenge-Response Pair (CRP) Mechanisms and Multi-Factor Authentication Schemes to
Protect Private Keys~\cite{ChallengeResponsePair2025a}, a scheme utilizing a CRP mechanism
to generate ephemeral keys for protecting private keys in distributed networks is presented.
The scheme incorporates a template-less biometric, an SRAM PUF, and a digital file,
combined to generate a cryptotable. This cryptotable can then be challenged to produce
responses, from which cryptographic ephemeral keys are derived. The uniqueness of this
approach lies in the ability to generate an effectively infinite number of keys on demand.

The work further extended this idea by using the generated keys as seeds for the public key
generation part of the CRYSTALS-Dilithium (ML-DSA)~\cite{CRYSTALSDilithium} digital
signature protocol, as well as to generate ephemeral keys that protect the private key
produced by the CRYSTALS-Dilithium digital signature scheme. Through the analysis, the
authors presented the distribution of 0, 1, and X generated from each reference factor, as
well as from the combined factors, i.e., cryptotable. In addition, they applied the
Response-Based Cryptographic (RBC) scheme~\cite{StatisticalAnalysisOptimize2020} to
evaluate key error rates under different fragmentation levels. In testing, user-specific
passwords were employed to demonstrate False Reject Rate (FRR) and False Acceptance Rate
(FAR) results across different fragmentation levels, validating the practical use case of the
scheme.

In this work, we performed a thorough statistical analysis of the proposed architecture to
optimize it for real-time use cases. Specifically, we introduced a novel bit-chopping scheme
of MSBs that not only strengthens security but also improves overall accuracy, as reflected
in the FAR and FRR results. Unlike the earlier study, our work relies on the same password
across all users and the same SRAM PUF-based token, thereby demonstrating the strength of
the system in discerning different individuals purely through the template-less biometric. We
intentionally excluded the digital file factor, as it consistently produced stable reference
tables, a finding already established in the prior work.

Furthermore, in this optimization study, we analyzed the template-less biometric scheme and
the SRAM-PUF token separately to identify the most effective configurations. For the
biometric scheme, we investigated the optimal accuracy-bit parameter, defined as the
Euclidean distance between a challenge point and a landmark point converted into a binary
representation, as well as the chopped-MSB parameter, defined as the truncation of the MSBs
from the accuracy-bit. For the SRAM-PUF token, we examined the number of enrollment cycles
required in order to determine the minimum cycle count necessary for a reliable enrollment
operation.

Beyond these contributions, several prior studies have also explored integrating PUFs with
biometrics under different assumptions and system
models~\cite{PUFBasedMutualAuthentication2023,UDhashingPhysicalUnclonable2019,SecureMutualAuthentication2021,PracticalSecureIoT2016,MutualAuthenticationIoT2017,LightweightHighlySecure2017,FacialBiohashingBased2018,SecureEfficientAKE2023}.
These schemes generally rely on server-side storage of challenge--response pairs or helper
data, assume error-free PUF responses, or require the reuse of server-provisioned secrets.
While these approaches highlight the potential of combining PUFs and biometrics, they often
suffer from practical limitations such as centralized vulnerabilities, error sensitivity, or
efficiency trade-offs. By contrast, our work eliminates server-side storage of CRPs,
introduces ephemeral key generation to ensure forward secrecy, and operates within a
zero-trust framework with mutual authentication, addressing many of the shortcomings
identified in earlier protocols.

%% ============================================================
\section{Methodology}
\label{sec:method}

We evaluated 10 SRAM-PUF tokens and 6,000 AI-generated images~\cite{GeneratedPhotos}. The
dataset comprises 400 individuals with 25 images each: 10 low-variation moment images and 15
high-variation images. From this, 200 individuals were selected; 100 were used for
enrollment and 100 for testing.

For the False Acceptance Rate (FAR), enrollment tables from the 100 enrolled individuals
were compared against single frames from the 15 variation images of the other 100, producing
150,000 tests. For the False Reject Rate (FRR), each individual's reference table was
compared with its own 10 variation images, yielding 15,000 tests. FRR experiments were
repeated 10 times to match the FAR test count, enabling direct comparison in
Figure~\ref{fig:histograms}.

For SRAM-PUF evaluation, each device produced 202 binary reads (Figure~\ref{fig:puf_error}),
from which 201 ternary enrollment files were derived. One hundred enrollments per token were
selected. In each test, 256 addresses from an enrollment file were compared against
subsequent binary reads over 1,000 trials. The procedure was repeated across 100 subsequent
reads, and results were averaged per enrollment. Figure~\ref{fig:puf_error} shows the
average key error per file.

The test script followed two phases: enrollment and ephemeral key generation. Enrollment was
executed independently for SRAM-PUF tokens and biometric data, as detailed in
Algorithms~\ref{alg:sram} and~\ref{alg:bio}.

We further analyzed the impact of ADC precision (4--8 bits) and MSB removal (0--4 bits) on
the biometric system, and the minimum enrollment time for SRAM-PUFs. Results are shown in
Figures~\ref{fig:histograms} and~\ref{fig:puf_error}.

Finally, we evaluated the combined system. Table~\ref{tab:top5} summarizes the
best-performing configurations, and Figure~\ref{fig:keybias} confirms that the generated
keys exhibit no detectable bias. All experiments were conducted on an
Intel\textsuperscript{\textregistered} Core\texttrademark\ i9-9900K CPU @ 3.60\,GHz with
32\,GB RAM and an NVIDIA RTX 3070 GPU.

\subsection{Enrollment Phase}

\subsubsection{SRAM PUF Token}

During the enrollment phase of the SRAM-PUF token, repeated power-cycling of the device
yields multiple million-bit SRAM dumps. From these dumps, a password- and RN1-seeded
generator deterministically selects $256\times256$ cell indices. Through repeated cycles, we
construct a ternary table $T \in \{0,1,\mathrm{X}\}^{256\times256}$.

Each individual read produces a binary table; the final ternary table is obtained by
superimposing multiple reads. If a specific cell position consistently retains the same value
across all reads, it is considered a stable cell, and its bit value is preserved.
Conversely, if flips are observed in sequential reads at a given cell position, that position
is marked as X, denoting a flaky cell. The complete process is described in
Algorithm~\ref{alg:sram}.

\subsubsection{Template-less Biometrics}

The enrollment process for template-less biometrics is based on a challenge--response
mechanism. Facial landmarks are detected on enrollment frames using
Dlib~\cite{DlibmlMachineLearning}. Randomly generated challenge points are placed on each
frame, and Euclidean distances between landmarks and challenge points are computed. These
distances form the raw measurements that are transformed into binary responses.

To improve stability, we apply Gray coding and introduce a bit-chopping step: MSBs of each
distance measurement are removed before conversion. MSB chopping reduces overlap between
biometric responses by filtering out large but uninformative variations, while preserving
lower-order bits that encode subject-specific distinctions.

Finally, 64 landmark indices are deterministically selected from a password-derived seed.
From multiple frames, the results are superimposed into a $256\times256$ table: stable
positions retain their binary values, while unstable ones are marked as X. In
Algorithm~\ref{alg:bio}, we present the step-by-step process that demonstrates how this
template-less biometric scheme operates in practice.

\begin{algorithm}[t]
\caption{SRAM PUF Table Construction}
\label{alg:sram}
\begin{algorithmic}[1]
\renewcommand{\algorithmicrequire}{\textbf{Input:}}
\renewcommand{\algorithmicensure}{\textbf{Output:}}
\REQUIRE Cycles $N$; serial port $P$; password \textit{pwd}; $RN1 \in \{0,1\}^{512}$;
         PUF size $M=1{,}048{,}576$ bits ($=131{,}072$ bytes); table size $C=256\times256$
\ENSURE $T \in \{0,1,\mathrm{X}\}^{256\times256}$
\STATE \textbf{Acquire million-bit dumps:}
\FOR{$k \leftarrow 1$ \TO $N$}
  \STATE Power on; wait $t_{\text{settle}}$; open $P$; trigger read;
  \STATE Read exactly 131,072 bytes $\to$ bitstring $R^{(k)} \in \{0,1\}^M$;
  \STATE Power off; wait $t_{\text{sleep}}$
\ENDFOR
\STATE \textbf{Derive keyed addresses ($RN1 \| \text{password}$):}
\STATE $RN1 \leftarrow 512$ uniform random bits;
\STATE $\textit{pwd}_{512} \leftarrow \text{SHA3-512}(\textit{pwd})$;
\STATE Byte stream $z \leftarrow \text{SHAKE-256}(RN1 \| \textit{pwd}_{512})$;
\STATE $S \leftarrow \emptyset$;
\WHILE{$|S| < C$}
  \STATE $u \leftarrow$ next uint32 from $z$; \quad $s \leftarrow u \bmod M$;
  \IF{$s \notin S$} \STATE insert $s$ into $S$ \ENDIF
\ENDWHILE
\STATE Let $S=[s_0,\ldots,s_{C-1}]$ be the ordered list of selected cell indices.
\STATE \textbf{Bit stability at selected cells:}
\FOR{$i \leftarrow 0$ \TO $C-1$}
  \IF{$R^{(1)}[s_i] = R^{(2)}[s_i] = \cdots = R^{(N)}[s_i]$}
    \STATE $U[i] \leftarrow R^{(1)}[s_i]$
  \ELSE
    \STATE $U[i] \leftarrow \mathrm{X}$
  \ENDIF
\ENDFOR
\STATE Generate the Token Table $T$.
\end{algorithmic}
\end{algorithm}

\begin{algorithm}[t]
\caption{BioTable Construction}
\label{alg:bio}
\begin{algorithmic}[1]
\renewcommand{\algorithmicrequire}{\textbf{Input:}}
\renewcommand{\algorithmicensure}{\textbf{Output:}}
\REQUIRE Password \textit{pwd}; $RN1\in\{0,1\}^{512}$; frames $N$; frame size $F=256$;
         table size $C=F^2$; accuracy-bits $g$; chopped-MSBs $m$
\ENSURE Ternary BioTable $B \in \{0,1,\mathrm{X}\}^{F\times F}$
\STATE \textbf{Face alignment \& landmarks:} Obtain aligned face chips of size $F\times F$;
       detect 68 landmarks per chip; seed RNG with SHA256(\textit{pwd}) and select 64
       unique landmark indices.
\STATE \textbf{Challenge generation:} Derive $C$ pseudorandom challenge coordinates in
       $[0,F)\times[0,F)$ from $\text{SHAKE-256}(RN1\|\text{SHA3-512}(\textit{pwd}))$.
\STATE \textbf{Per-frame distance features:} For each frame $f$ and challenge $q_i$,
       compute Euclidean distances to 64 selected landmarks:
       $d_i^{(f)}\in\mathbb{R}^{64}$.
\STATE \textbf{Quantization to Gray code:} $D_{\max} \leftarrow \sqrt{2}\,F$;
\FOR{each frame $f$ and challenge $i$}
  \FOR{$j=1:64$}
    \STATE $w \leftarrow \bigl\lfloor \tfrac{d_i^{(f)}[j]}{D_{\max}}\cdot 2^g \bigr\rfloor$;
    \STATE Take $g$-bit binary of $w$, chop $m$ MSBs; convert remaining $(g{-}m)$ bits
           to Gray code; concatenate over $j$ to get bitarray $R_i^{(f)}$.
  \ENDFOR
\ENDFOR
\STATE \textbf{Multi-frame unanimity $\to$ ternary map:}
       Initialize $U\in\{0,1,\mathrm{X}\}^C$;
\FOR{each bit index $i=1:C$}
  \IF{all frames agree at bit $i$}
    \STATE $U[i] \leftarrow$ unanimous bit
  \ELSE
    \STATE $U[i] \leftarrow \mathrm{X}$
  \ENDIF
\ENDFOR
\STATE Construct BioTable $T$.
\end{algorithmic}
\end{algorithm}

\subsection{Ephemeral Key Generation}

The key generation process for both the server and client follows the same structure;
however, unlike the enrollment phase, only a single read is used at this stage. The server
transmits the ternary table information to the client along with $RN1$ and $RN2$. On the
server side, a template-less format combining the SRAM-PUF template-less biometric tables is
employed, and the final key is derived using $RN2$ together with the user password. The
client mirrors this procedure by generating binary tables from one-shot reads of the
SRAM-PUF token and facial biometric, while incorporating the ternary table received from the
server. Using the same $RN1$, $RN2$, and password ensures alignment with the server's
process. After constructing its key, the client applies Response-Based Cryptography (RBC) to
guarantee consistency with the server's reconstruction, while the server compares its own
generated key against the fragmented, hashed key received from the client, applying RBC
error correction~\cite{StatisticalAnalysisOptimize2020} to mitigate real-time noise. Through
this protocol, both parties converge on an identical key, which can then serve as a secure
seed for post-quantum cryptographic schemes or be directly applied for encryption and
decryption to protect private keys. To simplify the testing process, both the server and
client were executed within the same script. The complete procedure of ephemeral key
generation is illustrated in Figure~\ref{fig:protocol}.

\begin{figure*}[!t]
\centering
\scalebox{0.85}{%
\begin{tikzpicture}[
  >=Latex,
  node distance = 10mm and 26mm,
  every node/.style={font=\footnotesize\sffamily},
  boxnode/.style={draw, rounded corners, align=left, inner sep=4pt,
                  text width=7.0cm, minimum height=6mm},
  note/.style={align=left, inner sep=0pt},
  msg/.style={midway, above=4pt, fill=white, rounded corners=1pt,
              inner sep=1pt, font=\sffamily\footnotesize, text opacity=1}
]
\coordinate (srvx) at (-6.0,0);
\coordinate (clix) at ( 6.0,0);
\node[note] (srvTitle) at (srvx) {\large\bfseries Server};
\node[note] (cliTitle) at (clix) {\large\bfseries Client};
\node[boxnode, below=7mm of srvTitle] (srv0)
  {Enrollment setup: define fixed $C{=}F^2$ indices (cryptotable); parameters:
   $g$ = distance bits, $m$ = chopped MSBs (bit truncation),
   $Q$ = oversampled challenged numbers, $L$ = 256 bit key.};
\node[boxnode, below=7mm of cliTitle] (cli0)
  {Enrollment setup: client knows how to build $T',B'$ with same
   $(g,m,Q,L)$ and fixed indices.};
\node[boxnode, below=of srv0] (srv1)
  {Generate $RN1$ (enrollment) and $RN2$ (challenges).
   Merge $T$ and $B$ into composite $C$; compute ternary map $F$.};
\node[boxnode, below=of cli0] (cli1)
  {Prepare to construct one-shot $T',B'$ once $(RN1,RN2,F)$ are received.};
\draw[->, thick, shorten >=2pt, shorten <=2pt]
  (srv1.east) to[out=5, in=175, looseness=1.05]
  node[msg]{\textbf{send} $(RN1, RN2, F)$}
  (cli1.west);
\node[boxnode, below=10mm of srv1] (srv2)
  {Derive index sequence from $(RN2\!\Vert\!\mathsf{pwd})$:
   SHA3-512($\mathsf{pwd}$) $\rightarrow$ 512-bit digest;
   seed = SHAKE-256($RN2\Vert$digest);
   generate challenge indices $a_i \in [0,C{-}1]$.};
\node[boxnode, below=10mm of cli1] (cli2)
  {Use $(RN1,RN2,F)$ to build $T',B'$.
   Apply $F$: if $F_j{=}1$ mark $X$, else $T'_j \oplus B'_j$.
   Derive identical challenge index stream $a_i$.};
\node[boxnode, below=of srv2] (srv3)
  {Collect stable set $\mathcal{S}$ until at least $Q$ bits.
   Abort if $|\mathcal{S}|<L$.
   ($Q$ ensures oversampling; $L$ is usable key length.)};
\node[boxnode, below=of cli2] (cli3)
  {Collect stable set $\mathcal{S}'$ until at least $Q$ bits.
   Fail if $|\mathcal{S}'|<L$.};
\node[boxnode, below=of srv3] (srv4)
  {Select $L$ distinct indices (no replacement) from $\mathcal{S}$;
   concatenate $\Rightarrow$ server key $\mathbf{K}$.};
\node[boxnode, below=of cli3] (cli4)
  {Select $L$ distinct indices (no replacement) from $\mathcal{S}'$;
   concatenate $\Rightarrow$ client key $\mathbf{K'}$.};
\node[boxnode, below=of cli4] (cli5)
  {Compute fragmented hash $\mathbf{H_{\mathrm{frag}}(K')}$; prepare to send.};
\node[boxnode, below=of srv4] (srv6)
  {Verify received $\mathbf{H_{\mathrm{frag}}(K')}$ against server key $\mathbf{K}$;
   apply RBC to reconcile residual noise; finalize $\mathbf{K}$.};
\draw[->, thick, shorten >=2pt, shorten <=2pt]
  (cli5.west) to[out=185, in=-5, looseness=1.05]
  node[msg]{\textbf{send} $H_{\mathrm{frag}}(K')$}
  (srv6.east);
\foreach \i/\j in {0/1,1/2,2/3,3/4,4/6}
  { \draw[->] (srv\i.south) -- (srv\j.north); }
\foreach \i/\j in {0/1,1/2,2/3,3/4,4/5}
  { \draw[->] (cli\i.south) -- (cli\j.north); }
\end{tikzpicture}%
}
\caption{Protocol diagram for ephemeral key generation.}
\label{fig:protocol}
\end{figure*}
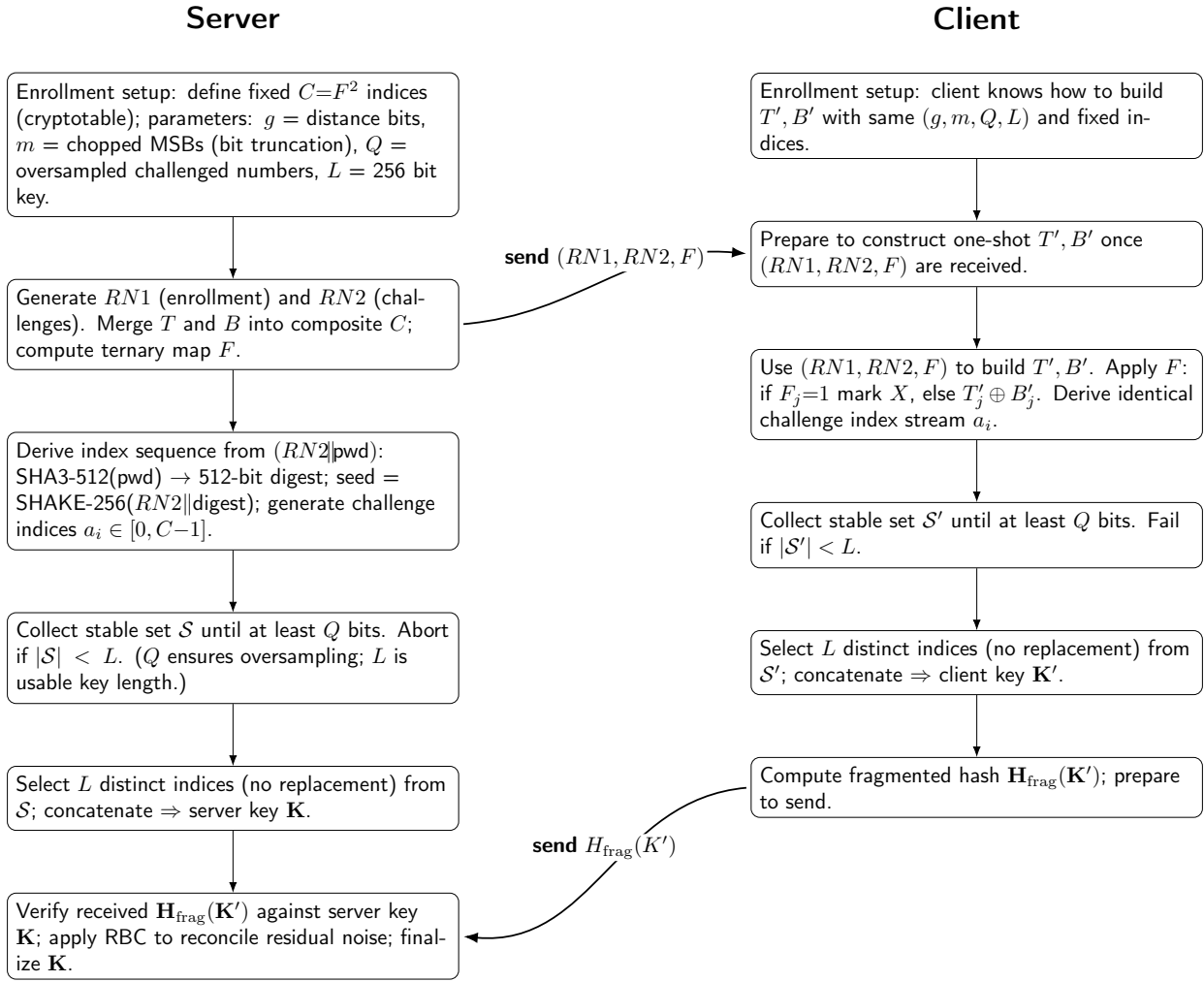

%% ============================================================
\section{Experimental Evaluation}
\label{sec:eval}

\subsection{Template-less Biometrics Analysis}

We tested a total of 15 configurations; however, in Figure~\ref{fig:histograms} we only
present the results for accuracy-bit and chopped-MSB configurations of (6,1), (6,2), (7,1),
(7,2), (8,1), and (8,2). Our experimental analysis shows that the gap between correct and
incorrect individuals increases as the number of distance bits grows, and further increases
with a larger number of chopped MSBs. While increasing accuracy-bits and chopped-MSBs has a
positive effect on differentiating between users---in other words, reduced FAR---we later
observed that it can also raise the FRR. Some overlap between correct and incorrect
individuals is visible, which occurs because the protocol is designed for straight-face
images and is less tolerant to angle variations. To evaluate these limitations, we performed
enrollment primarily on straight-face images and then conducted key generation using the
variation bucket of images. Our future goal is to design a protocol that supports 3D facial
data, which will eliminate these overlaps and improve robustness.

\begin{figure*}[t]
  \centering
  \begin{subfigure}[b]{0.32\textwidth}
    \includegraphics[width=\textwidth]{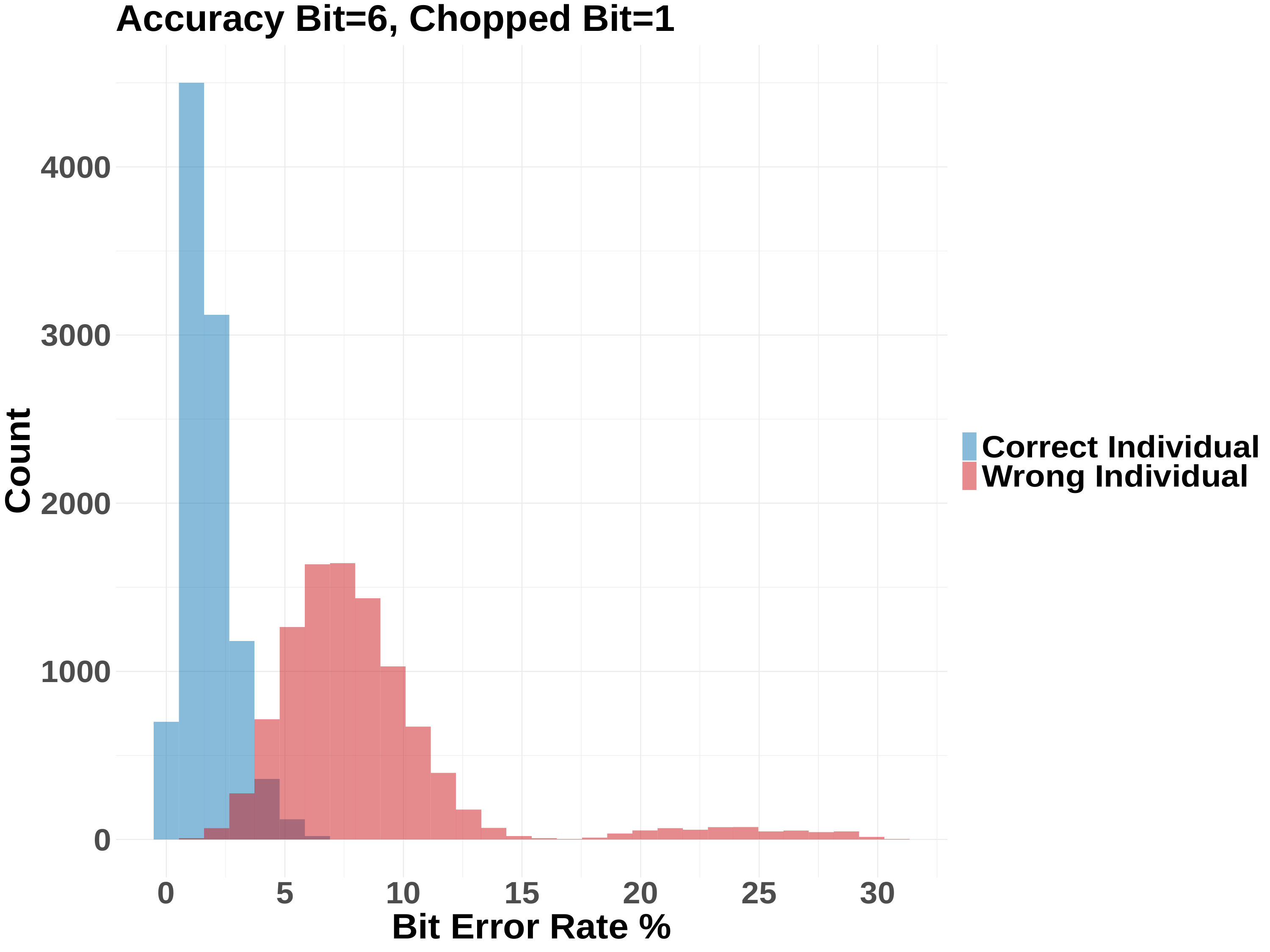}
    \caption{Accuracy Bit=6, Chopped Bit=1}\label{fig:h61}
  \end{subfigure}\hfill
  \begin{subfigure}[b]{0.32\textwidth}
    \includegraphics[width=\textwidth]{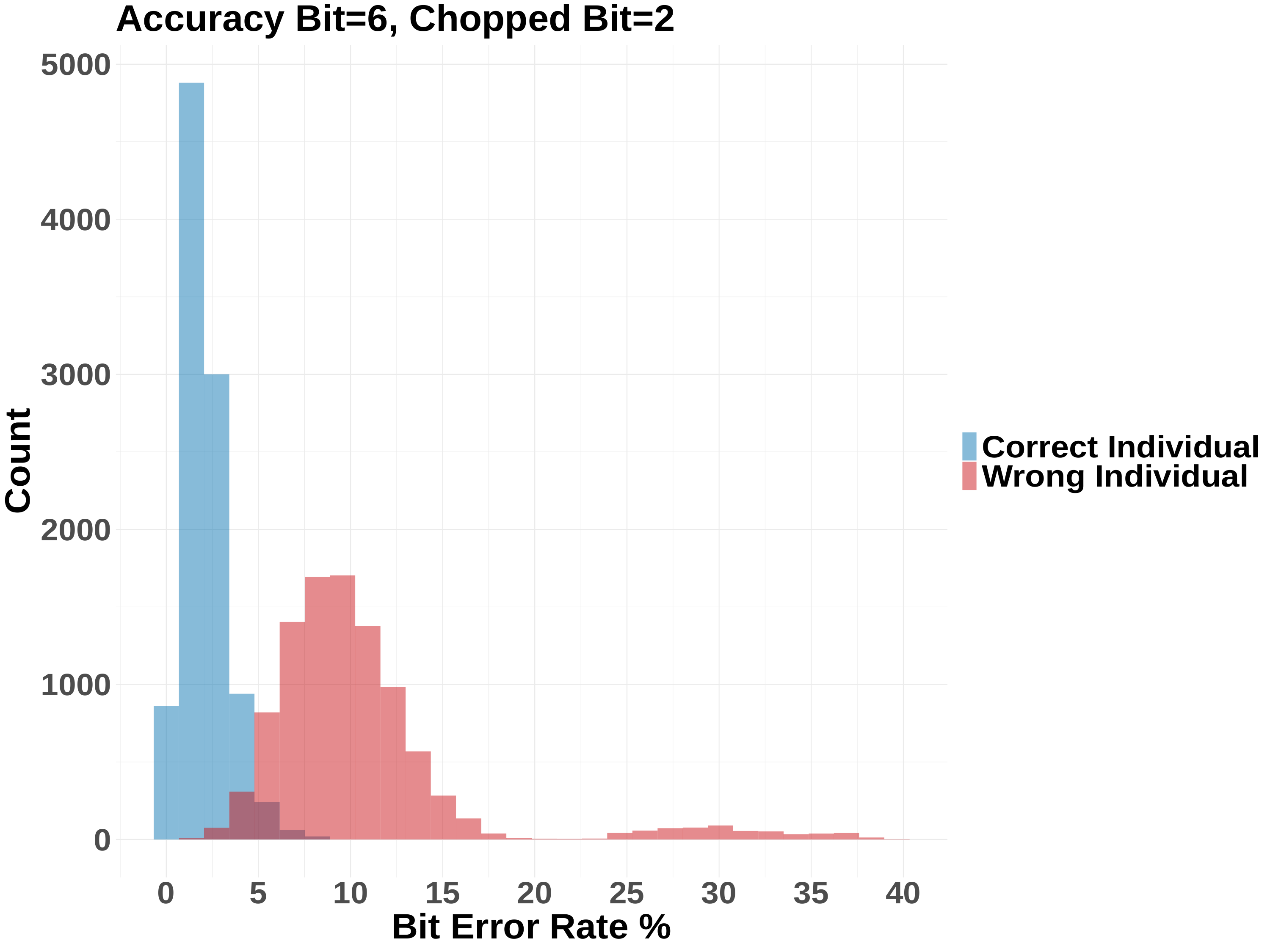}
    \caption{Accuracy Bit=6, Chopped Bit=2}\label{fig:h62}
  \end{subfigure}\hfill
  \begin{subfigure}[b]{0.32\textwidth}
    \includegraphics[width=\textwidth]{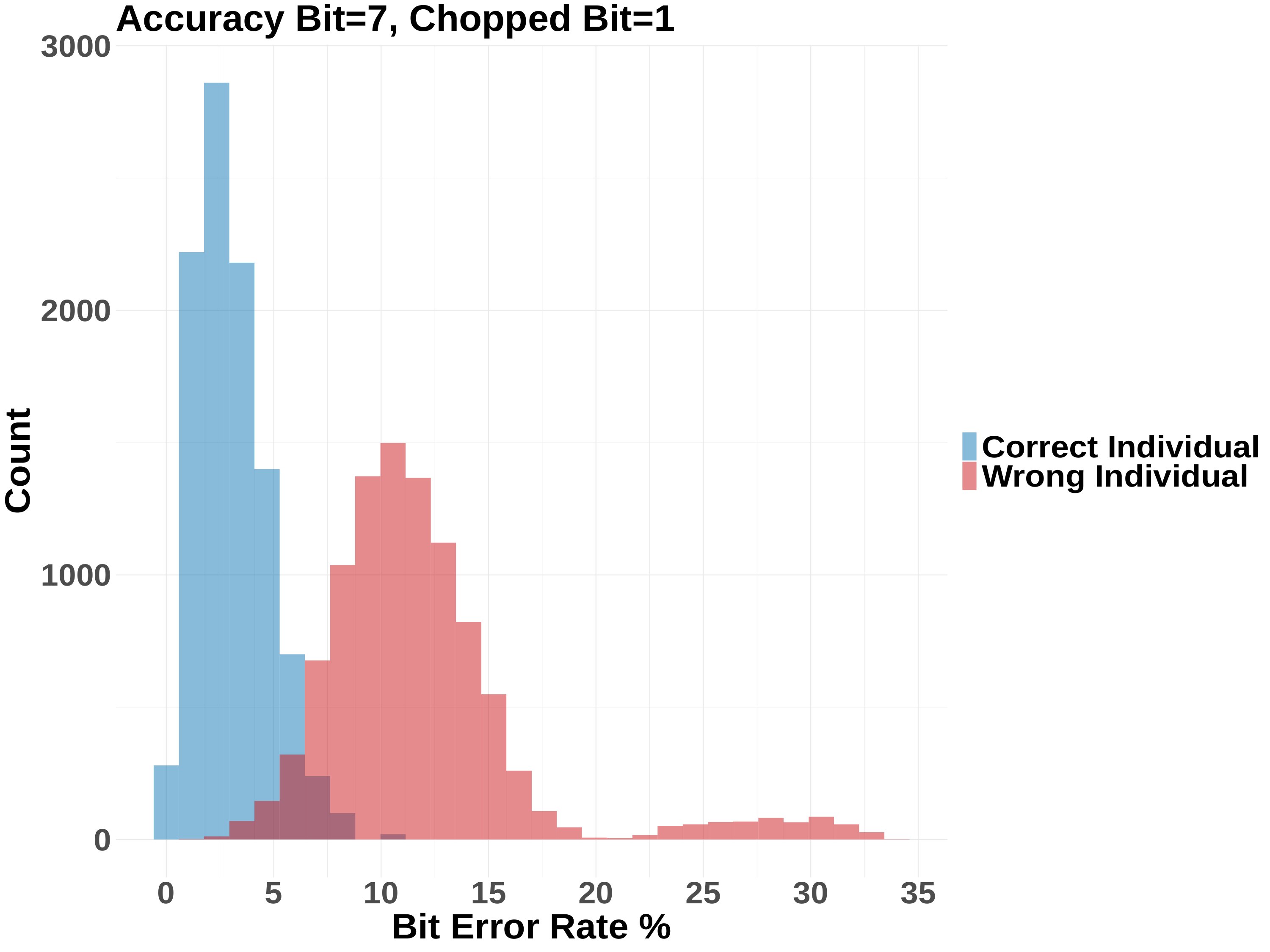}
    \caption{Accuracy Bit=7, Chopped Bit=1}\label{fig:h71}
  \end{subfigure}\\[6pt]
  \begin{subfigure}[b]{0.32\textwidth}
    \includegraphics[width=\textwidth]{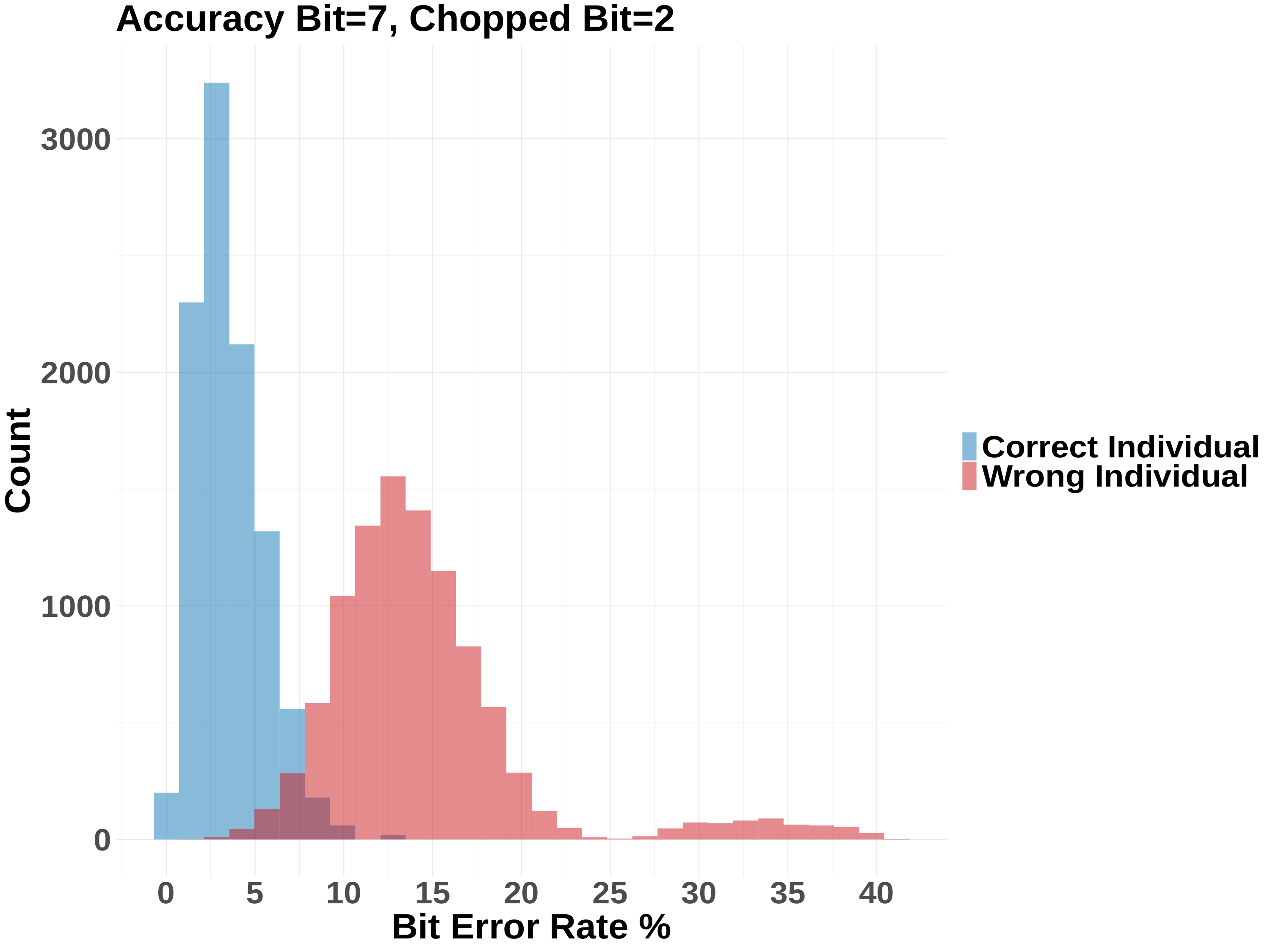}
    \caption{Accuracy Bit=7, Chopped Bit=2}\label{fig:h72}
  \end{subfigure}\hfill
  \begin{subfigure}[b]{0.32\textwidth}
    \includegraphics[width=\textwidth]{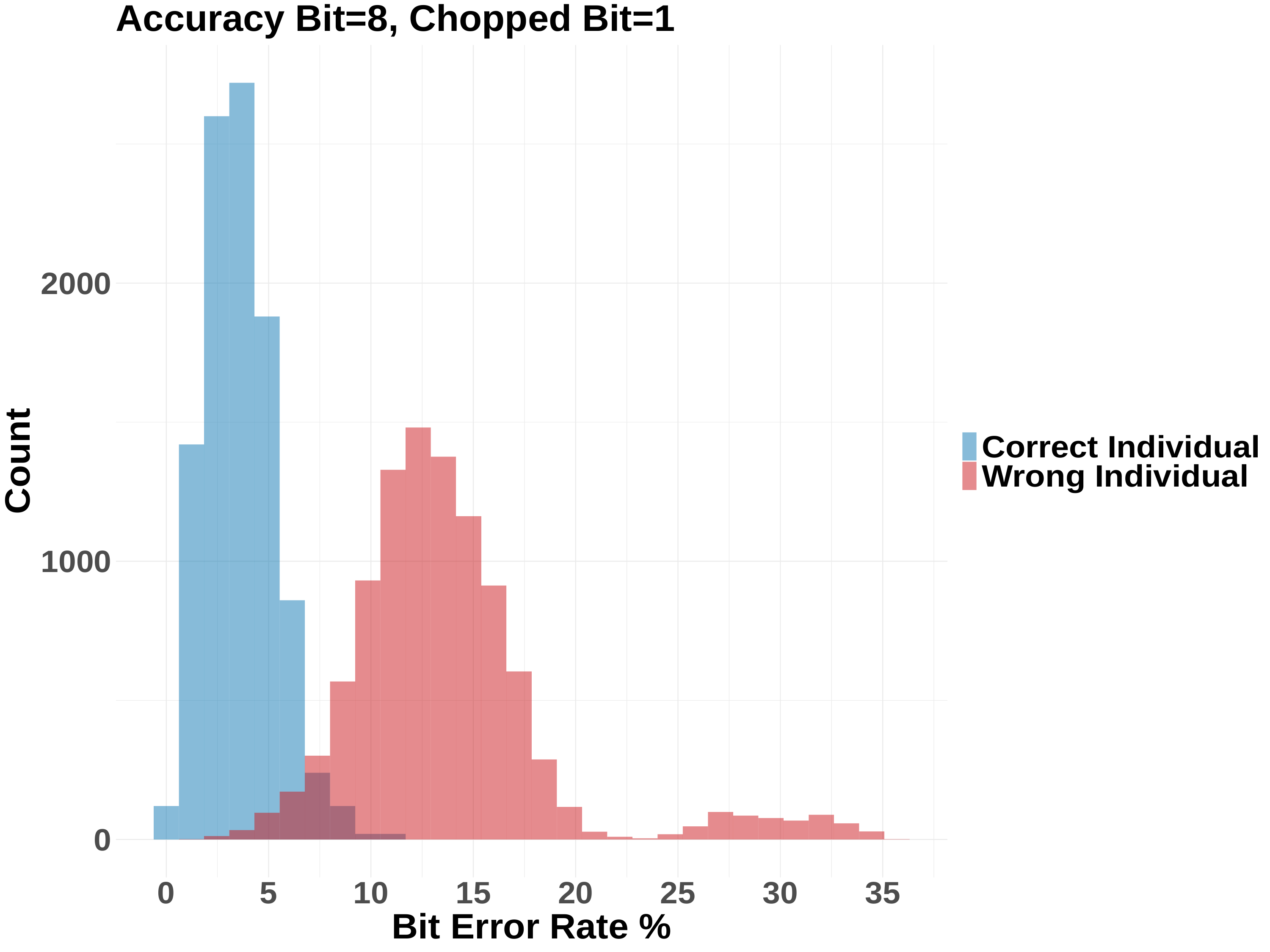}
    \caption{Accuracy Bit=8, Chopped Bit=1}\label{fig:h81}
  \end{subfigure}\hfill
  \begin{subfigure}[b]{0.32\textwidth}
    \includegraphics[width=\textwidth]{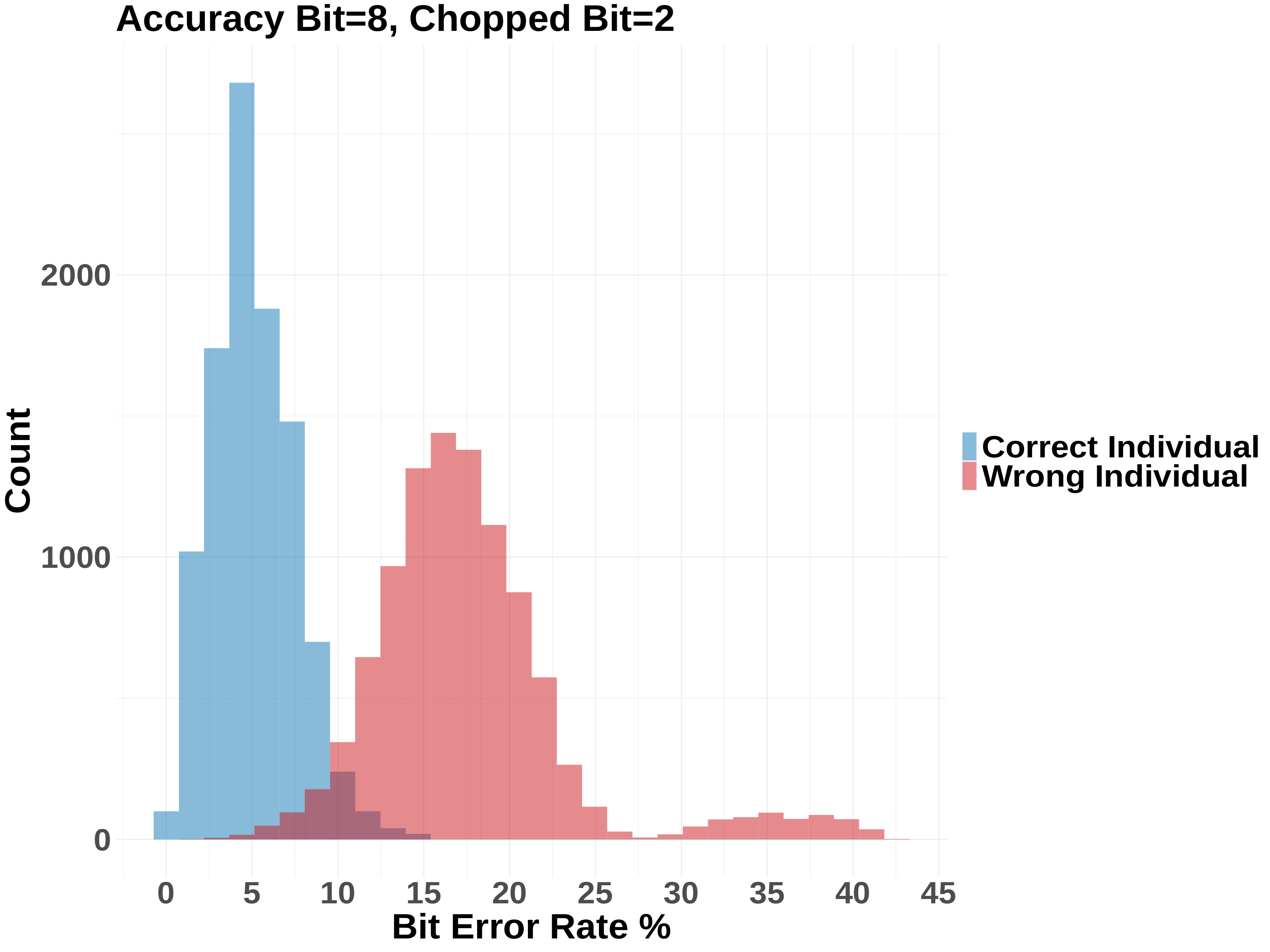}
    \caption{Accuracy Bit=8, Chopped Bit=2}\label{fig:h82}
  \end{subfigure}
  \caption{Comparison of histograms for different configurations for template-less biometrics
           based on (Distance bits, Chopped MSB).}
  \label{fig:histograms}
\end{figure*}

\subsection{SRAM PUF Token Analysis}

In Figure~\ref{fig:puf_error}, we analyze the average key error versus the number of
enrollments. As the number of enrollments increases, the average error clearly decreases.
This behavior arises from a fundamental property of SRAM PUFs: with more reads, unstable
cells can be more effectively identified and flagged. However, the practical limit of
enrollment must be considered, since increasing enrollment time also increases the overall
cost of the process. Our experiments show that 20 enrollments are sufficient to generate
keys with limited error, and the RBC error correction scheme can reliably handle this level
of noise.

\begin{figure}[h]
  \centering
  \includegraphics[width=0.7\linewidth]{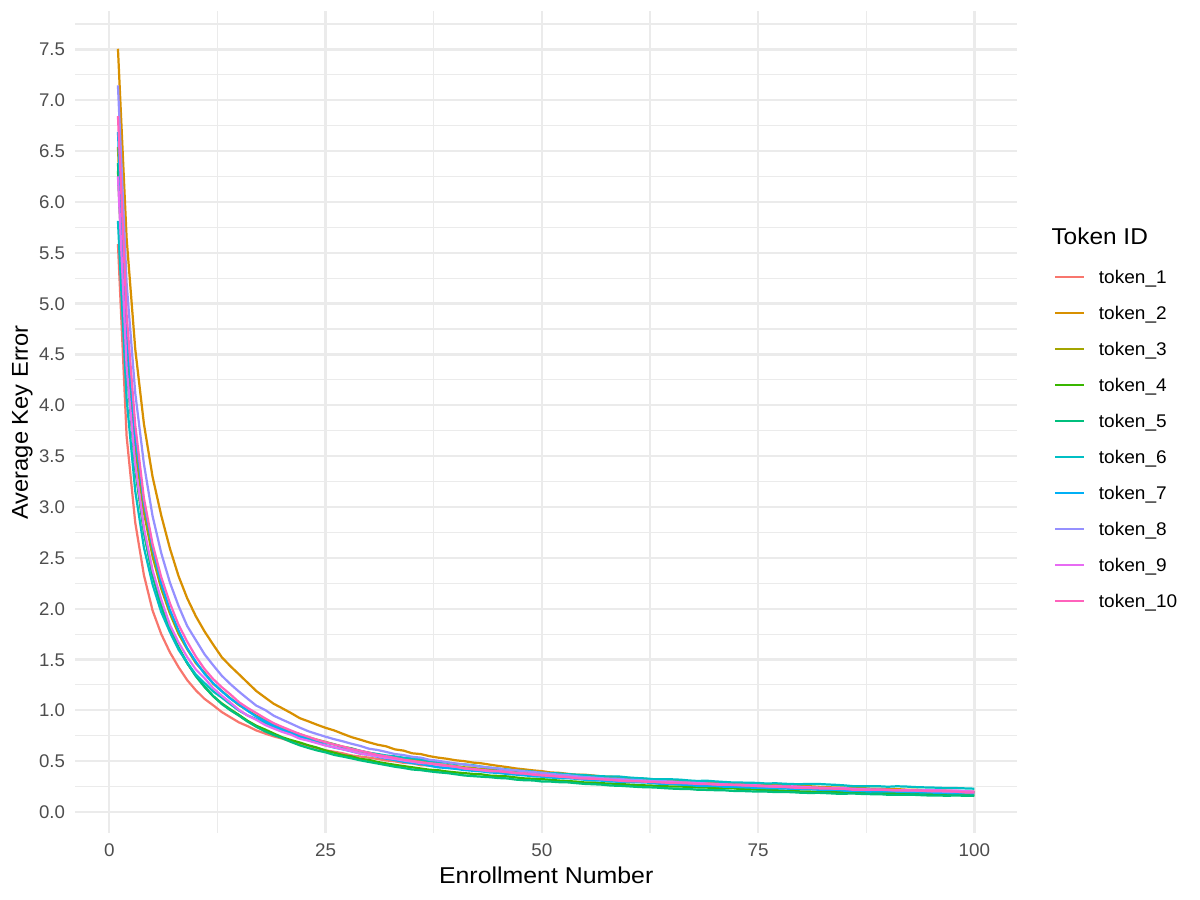}
  \caption{Number of key errors reduces as we increase the number of enrollments for SRAM
           PUF Token.}
  \label{fig:puf_error}
\end{figure}

\subsection{Key Bias Test}

To evaluate the balance between zeros and ones in the generated ephemeral keys, we conducted
exact binomial tests under the null hypothesis of an unbiased bitstream ($p(1)=0.5$), as the
results show in Figure~\ref{fig:keybias}. For the client pool ($n=115{,}200$), the observed
proportion of ones was $\hat{p}=0.4998$ with a 95\% confidence interval of
$[0.4969,\,0.5027]$, yielding a non-significant result ($p=0.9085$). Similarly, the server
pool ($n=115{,}200$) produced $\hat{p}=0.4999$, 95\% CI $[0.4970,\,0.5027]$, with
$p=0.9225$. In both cases, the confidence intervals tightly bound 0.5, and the high
$p$-values indicate no evidence of systematic bias in the relative frequencies of zeros and
ones.

\begin{figure}[h]
  \centering
  \includegraphics[width=0.7\linewidth]{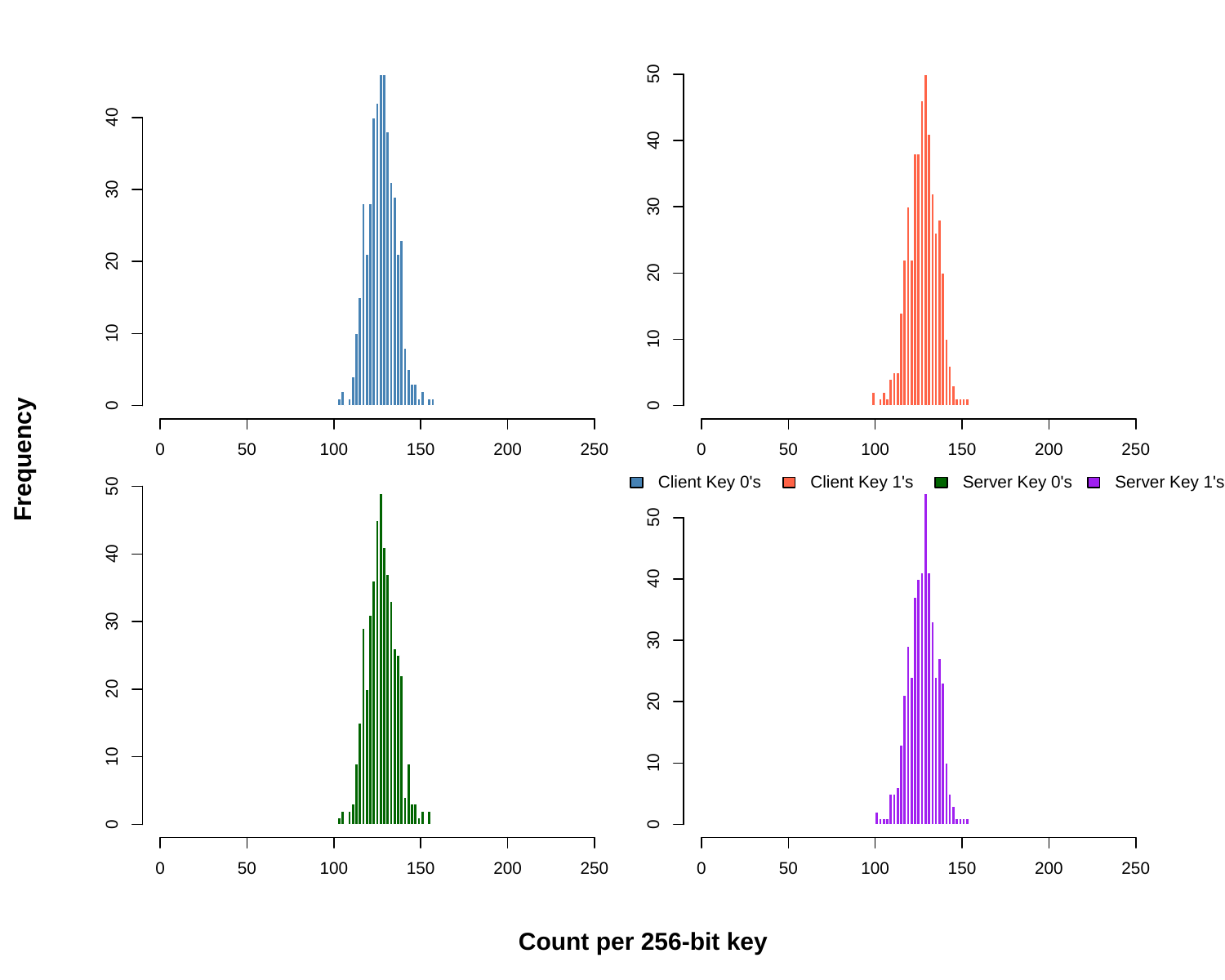}
  \caption{No bias in the generated key.}
  \label{fig:keybias}
\end{figure}

\subsection{Overall Result}

As shown in the top five configuration results in Table~\ref{tab:top5}, our technology
demonstrates promising potential, achieving 0\% FAR and 0\% FRR. This test was performed
using a random sampling of $n=30$ from the dataset. For this evaluation, we assumed that all
users share the same password and the same token, with enrollment performed over 40 cycles.
On average, each enrollment cycle takes about 5 seconds, resulting in a total enrollment time
of less than 3.5 minutes. However, we later determined that 20 enrollment cycles are
sufficient, as the number of fragmentations in RBC can be increased to tolerate one or two
additional bit flips caused by the reduced number of enrollments.

\begin{table}[h]
\centering
\caption{Top 5 configurations with lowest error rates.}
\label{tab:top5}
\rowcolors{2}{BestGreen}{white}
\begin{tabular}{ccccc}
\toprule
\rowcolor{HeaderBlue}
\makecell{\textbf{Frag.}\\\textbf{Level}} &
\makecell{\textbf{Accuracy}\\\textbf{Bits}} &
\makecell{\textbf{Chopped}\\\textbf{MSB}} &
\textbf{FAR (\%)} &
\textbf{FRR (\%)} \\
\midrule
1 & 7 & 1 & \textbf{0.0} & \textbf{0.0} \\
1 & 7 & 2 & \textbf{0.0} & \textbf{0.0} \\
2 & 7 & 2 & \textbf{0.0} & \textbf{0.0} \\
1 & 6 & 2 & 0.0           & 1.7 \\
4 & 7 & 2 & 3.3           & 0.0 \\
\bottomrule
\end{tabular}
\end{table}

%% ============================================================
\section{Security Analysis}
\label{sec:security}

\subsection{Resistance to Sequential Attacks}

Since the authentication template is a combination of template-less biometrics and the
SRAM-PUF token, sequential attacks cannot be mounted successfully. The dynamic and fused
nature of the template prevents adversaries from exploiting repeated authentication attempts.

\subsection{Resistance to Insider Attacks}

The stored template is one-way and unlinkable. Even in the event of insider access, the
protected credentials do not reveal usable information about the user's biometric features or
PUF responses.

\subsection{Resistance to Man-in-the-Middle Attacks}

Although an adversary may intercept random nonces ($RN1$, $RN2$) during transmission, these
values alone are insufficient to reconstruct the authentication key. Without access to the
actual PUF responses and facial biometric data, the intercepted information has no utility.

\subsection{Resistance to Modeling Attacks}

Modeling attacks are ineffective because the system does not rely on static
challenge--response pairs. Instead, it generates a template-less cryptotable, making it
infeasible for an attacker to reconstruct or predict the key through machine learning or
statistical modeling.

\subsection{Mathematical One-wayness of the Cryptotable}

\begin{theorem}
Let $f$ be the cellwise merge operator defined by
\[
D_{ij} =
\begin{cases}
\mathrm{X}, & \text{if } A_{ij}=\mathrm{X} \text{ or } B_{ij}=\mathrm{X},\\
A_{ij}\oplus B_{ij}\oplus C_{ij}, & \text{if } A_{ij},B_{ij}\in\{0,1\},
\end{cases}
\]
with $A,B\in\{0,1,\mathrm{X}\}^{n\times n}$ and $C\in\{0,1\}^{n\times n}$.
Then $f$ is not injective.
\end{theorem}

\begin{proof}[Proof 1 (Counting Argument)]
Consider a cell $(i,j)$ with $D_{ij}\in\{0,1\}$. The triple $(A_{ij},B_{ij},C_{ij})$ must
satisfy the linear constraint $A_{ij}\oplus B_{ij}\oplus C_{ij}=D_{ij}$. This is one parity
equation over three binary variables, leaving two degrees of freedom. Thus there are exactly
4 valid preimages for each such output cell. If $m$ cells of $D$ take binary values, the
total number of preimages is $4^m$. Cells with $D_{ij}=\mathrm{X}$ admit even more
preimages, since at least one of $A_{ij}$ or $B_{ij}$ can be freely chosen as X. Hence the
map $f$ is highly many-to-one.
\end{proof}

\begin{proof}[Proof 2 (Explicit Collision)]
Fix a cell $(i,j)$ such that $D_{ij}\in\{0,1\}$. Suppose
$(A_{ij},B_{ij},C_{ij})=(0,0,C_{ij})$ yields $D_{ij}=C_{ij}$. Then the alternative
assignment $(A_{ij},B_{ij},C_{ij})=(1,1,C_{ij})$ also yields the same output, since
$1\oplus 1=0$ and thus $1\oplus 1\oplus C_{ij}=C_{ij}$. Therefore, by modifying $(A,B)$ at
a single cell from $(0,0)$ to $(1,1)$ (keeping $C$ fixed), one obtains a distinct input
triple with identical output $D$. This exhibits a collision and proves that $f$ is not
injective.
\end{proof}

%% ============================================================
\section{Future Work}
\label{sec:future}

This work establishes the groundwork for a practical, application-level zero-knowledge MFA
technology to protect private keys in a distributed network. There remain several avenues for
improvement. One direction is the addition of a sensor-based PUF factor, which is currently
under development and will later be embedded in wearable devices. Another direction is the
design of a challenge--response protocol using 3D biometrics to enhance the robustness of
the current 2D-based system. Furthermore, conducting case studies---such as the protection
of medical records or the integration with blockchain technology---could demonstrate the
applicability of this approach and open new research directions. Finally, a thorough security
analysis of the proposed architecture is currently underway.

%% ============================================================
\section{Conclusion}
\label{sec:conclusion}

This work demonstrates that the proposed bit-chopping of MSBs significantly improves both
the accuracy and security of the architecture. Statistical analysis enabled the optimization
of template-less biometry and SRAM-PUF--based key generation. In the template-less biometric
scheme, insufficient distance accuracy leads to high false-acceptance rates, while excessive
accuracy increases false-reject rates. Analysis of 4--8 ADC precision bits with 0--4 MSBs
removed indicates that 6--7 bits of precision with 1--2 MSBs removed achieve the best
balance. Furthermore, eliminating MSBs reduces false acceptances without substantially
affecting legitimate users.

For the SRAM-PUF scheme, approximately 20 power-off/power-on enrollment cycles are
sufficient to identify and eliminate unstable cells. Beyond this threshold, additional cycles
yield minimal improvement in bit error rates, allowing the enrollment process to remain short
and efficient.

Future work will focus on incorporating additional authentication factors and further
strengthening the security of the protocol.

%% ============================================================
\bibliographystyle{ieeetr}
\bibliography{ref}

@article{ChallengeResponsePair2025a,
  title     = {Challenge--{Response} {Pair} {Mechanisms} and {Multi-Factor} {Authentication} {Schemes} to {Protect} {Private} {Keys}},
  author    = {Cambou, Bertrand Francis and Alam, Mahafujul},
  journal   = {Applied Sciences},
  volume    = {15},
  number    = {6},
  pages     = {3089},
  year      = {2025},
  publisher = {MDPI},
  url       = {https://www.mdpi.com/2076-3417/15/6/3089}
}

@inproceedings{SiliconPhysicalRandom2002,
  title     = {Silicon Physical Random Functions},
  author    = {Gassend, Blaise and Clarke, Dwaine and van Dijk, Marten and Devadas, Srinivas},
  booktitle = {Proceedings of the 9th {ACM} Conference on Computer and Communications Security},
  series    = {{CCS}'02},
  pages     = {148--160},
  year      = {2002},
  publisher = {Association for Computing Machinery},
  doi       = {10.1145/586110.586132}
}

@article{herderPhysicalUnclonableFunctions2014,
  title   = {Physical {Unclonable} {Functions} and {Applications}: {A} Tutorial},
  author  = {Herder, Charles and Yu, Meng-Day and Koushanfar, Farinaz and Devadas, Srinivas},
  journal = {Proceedings of the IEEE},
  volume  = {102},
  number  = {8},
  pages   = {1126--1141},
  year    = {2014},
  doi     = {10.1109/JPROC.2014.2320516}
}

@article{CRYSTALSDilithium,
  title   = {{CRYSTALS-Dilithium}: A Lattice-Based Digital Signature Scheme},
  author  = {Bai, Shi and Ducas, L{\'e}o and Kiltz, Eike and Lepoint, Tancrèd and Lyubashevsky, Vadim and Schwabe, Peter and Seiler, Gregor and Stehl{\'e}, Damien},
  journal = {IACR Transactions on Cryptographic Hardware and Embedded Systems},
  year    = {2021}
}

@inproceedings{StatisticalAnalysisOptimize2020,
  title     = {Statistical {Analysis} to {Optimize} the {Generation} of {Cryptographic} {Keys} from {Physical} {Unclonable} {Functions}},
  author    = {Cambou, Bertrand and Mohammadi, Mohammad and Philabaum, Christopher and Booher, Duane},
  booktitle = {Intelligent Computing},
  editor    = {Arai, Kohei and Kapoor, Supriya and Bhatia, Rahul},
  pages     = {302--321},
  year      = {2020},
  publisher = {Springer International Publishing}
}

@article{PUFBasedMutualAuthentication2023,
  title   = {{PUF}-Based Mutual Authentication and Key Exchange Protocol for Peer-to-Peer {IoT} Applications},
  author  = {Zheng, Yue and Liu, Wenye and Gu, Chongyan and Chang, Chip-Hong},
  journal = {IEEE Transactions on Dependable and Secure Computing},
  volume  = {20},
  number  = {4},
  pages   = {3299--3316},
  year    = {2023},
  doi     = {10.1109/TDSC.2022.3193570}
}

@article{UDhashingPhysicalUnclonable2019,
  title   = {{UDhashing}: Physical Unclonable Function-Based User-Device Hash for Endpoint Authentication},
  author  = {Zheng, Yue and Cao, Yuan and Chang, Chip-Hong},
  journal = {IEEE Transactions on Industrial Electronics},
  volume  = {66},
  number  = {12},
  pages   = {9559--9570},
  year    = {2019},
  doi     = {10.1109/TIE.2019.2893831}
}

@inproceedings{SecureMutualAuthentication2021,
  title     = {Secure Mutual Authentication and Key-Exchange Protocol between {PUF}-Embedded {IoT} Endpoints},
  author    = {Zheng, Yue and Chang, Chip-Hong},
  booktitle = {2021 {IEEE} International Symposium on Circuits and Systems ({ISCAS})},
  pages     = {1--5},
  year      = {2021},
  doi       = {10.1109/ISCAS51556.2021.9401135}
}

@inproceedings{PracticalSecureIoT2016,
  title     = {Practical and Secure {IoT} Device Authentication Using Physical Unclonable Functions},
  author    = {Wallrabenstein, John Ross},
  booktitle = {2016 {IEEE} 4th International Conference on Future Internet of Things and Cloud ({FiCloud})},
  pages     = {99--106},
  year      = {2016},
  doi       = {10.1109/FiCloud.2016.22}
}

@article{MutualAuthenticationIoT2017,
  title   = {Mutual Authentication in {IoT} Systems Using Physical Unclonable Functions},
  author  = {Aman, Muhammad Naveed and Chua, Kee Chaing and Sikdar, Biplab},
  journal = {IEEE Internet of Things Journal},
  volume  = {4},
  number  = {5},
  pages   = {1327--1340},
  year    = {2017},
  doi     = {10.1109/JIOT.2017.2703088}
}

@inproceedings{LightweightHighlySecure2017,
  title     = {Lightweight Highly Secure {PUF} Protocol for Mutual Authentication and Secret Message Exchange},
  author    = {Idriss, Tarek and Bayoumi, Magdy},
  booktitle = {2017 {IEEE} International Conference on {RFID} Technology \& Application ({RFID-TA})},
  pages     = {214--219},
  year      = {2017},
  doi       = {10.1109/RFID-TA.2017.8098893}
}

@inproceedings{FacialBiohashingBased2018,
  title     = {Facial Biohashing Based User-Device Physical Unclonable Function for Bring Your Own Device Security},
  author    = {Zheng, Yue and Cao, Yuan and Chang, Chip-Hong},
  booktitle = {2018 {IEEE} International Conference on Consumer Electronics ({ICCE})},
  pages     = {1--6},
  year      = {2018},
  doi       = {10.1109/ICCE.2018.8326074}
}

@article{SecureEfficientAKE2023,
  title   = {A Secure and Efficient {AKE} Scheme for {IoT} Devices Using {PUF} and Cancellable Biometrics},
  author  = {Zahednejad, Behnam and Gao, Chong-zhi},
  journal = {Internet of Things},
  volume  = {24},
  pages   = {100937},
  year    = {2023},
  doi     = {10.1016/j.iot.2023.100937}
}

@misc{GeneratedPhotos,
  author       = {{Generated Photos}},
  title        = {Generated Photos: {AI} Generated Faces},
  howpublished = {\url{https://generated.photos/}},
  note         = {Accessed: 2025-08-29},
  year         = {2025}
}

@article{DlibmlMachineLearning,
  title   = {Dlib-ml: A Machine Learning Toolkit},
  author  = {King, Davis E.},
  journal = {Journal of Machine Learning Research},
  volume  = {10},
  pages   = {1755--1758},
  year    = {2009}
}

\end{document}